\documentclass[sigconf,authorversion]{acmart}

\usepackage{amsthm}

\makeatletter
\newtheorem*{rep@theorem}{\rep@title}
\newcommand{\newreptheorem}[2]{%
\newenvironment{rep#1}[1]{%
 \def\rep@title{#2 \ref{##1}}%
 \begin{rep@theorem}}%
 {\end{rep@theorem}}}
\makeatother

\newreptheorem{theorem}{Theorem}

\newtheorem{thm}{Theorem}[section]

\newtheorem{proposition}[thm]{Proposition}

\theoremstyle{definitionnition}
\newtheorem{definition}[thm]{Definition}

\theoremstyle{remark}

\AtBeginDocument{%
  \providecommand\BibTeX{{%
    \normalfont B\kern-0.5em{\scshape i\kern-0.25em b}\kern-0.8em\TeX}}}

\setcopyright{acmcopyright}
\copyrightyear{2018}
\acmYear{2018}
\acmDOI{10.1145/1122445.1122456}

\acmConference[AFT '20]{AFT '20:  ACM conference on Advances in Financial Technologies}{October 21--23, 2020}{New York City}
\acmBooktitle{AFT '20: ACM conference on Advances in Financial Technologies,
 October 21--23, 2020, New York City}
\acmPrice{15.00}
\acmISBN{978-1-4503-XXXX-X/18/06}

\begin{document}

\title{Resource Pools and the CAP Theorem}

\author{Andrew Lewis-Pye}
\email{a.lewis7@lse.ac.uk}
\affiliation{%
  \institution{Department of Mathematics, London School of Economics}
  \streetaddress{Columbia House, Houghton Street}
  \city{London}
  \postcode{WC2A 2AE}
}

\author{Tim Roughgarden}
\authornote{Supported in part by NSF Award
    CCF-1813188, ARO grant W911NF1910294, and a grant from the
 Columbia-IBM Center for Blockchain and Data Transparency.}
\email{tim.roughgarden@gmail.com}
\affiliation{%
  \institution{Columbia University}
  \streetaddress{500 West 120th Street, Room 450}
  \city{New York}
  \postcode{NY 10027}
}


\begin{abstract}
  Blockchain protocols differ in fundamental ways, including the
  mechanics of selecting users to produce blocks (e.g., proof-of-work
  vs.\ proof-of-stake) and the method to establish consensus (e.g.,
  longest chain  rules vs.\ BFT-inspired protocols).  These fundamental
  differences have hindered ``apples-to-apples'' comparisons between
  different categories of blockchain protocols and, in turn, the
  development of theory to formally discuss their relative merits.

  This paper presents a parsimonious abstraction sufficient for
  capturing and comparing properties of many well-known permissionless
  blockchain protocols, simultaneously capturing essential
  properties of both proof-of-work and proof-of-stake protocols, and
  of both longest-chain-type and BFT-type protocols.
Our framework blackboxes the precise
  mechanics of the user selection process, allowing us to isolate the
  properties of the selection process which are significant for
  protocol design.  

We illustrate our framework's utility with two
  results.  First, we prove an analog of the CAP theorem from
  distributed computing for our
  framework in a partially synchronous setting.
This theorem shows that a fundamental dichotomy holds
  between protocols (such as Bitcoin) that are \emph{adaptive}, in the
  sense that they
  can function given unpredictable levels of participation, and
  protocols (such as Algorand) that have certain \emph{finality}
  properties.  Second, we formalize the idea
that proof-of-work (PoW) protocols
  and non-PoW protocols can be distinguished by the forms of
  permission that users are given to carry out updates to the state.
\end{abstract}

\ccsdesc[500]{Theory of computation}

\keywords{blockchains, cryptocurrencies, distributed computing, CAP theorem}


\maketitle

\section{Introduction} \label{secIntro}

The task of a permissionless blockchain  protocol is to establish consensus for
message ordering over a network of users. This job
is made difficult by the fact that, being subject to the laws of
physics, the underlying communication network must have
\emph{latency}, i.e.\ broadcast  messages will
necessarily take time to travel over the network of users. As a
consequence of this latency, malicious users may purposely cause the
order in which messages are first seen to be different for
different honest users, and some differences in ordering will anyway
be an honest consequence of varying propagation times between
different nodes of the network \cite{decker2013information}.

Network latency is especially problematic when we work at the
level of individual transactions, which may be produced at a rate
which is high compared to network latency.  For this reason it is
standard practice to collect transactions together into \emph{blocks},
which can then be produced at a rate which is much lower compared to
network latency. Given that we are working in a permissionless
setting, the basic question then becomes, ``Who should produce the
blocks?''.  
This question is answered differently by different protocols.
Three running examples, which we shall use for reference in this paper, are:
\begin{itemize}

\item Bitcoin;

\item Snow White; 

\item Algorand. 

\end{itemize}

Of course, Bitcoin \cite{nakamoto2019bitcoin} is the best known
proof-of-work (PoW) protocol, and is also a \emph{longest chain}
protocol. This means that forks may occur in the blockchain, but that
honest miners will build on the longest chain. At a high level, Snow
White \cite{bentov2016snow} might be seen as a proof-of-stake (PoS)
version of Bitcoin -- it is also a longest chain protocol, but now
miners are selected with probability proportional to their stake in
the currency, rather than their hashing power. We use Algorand
\cite{chen2016algorand} as an example of a `BFT' protocol.\footnote{By
  `BFT protocols' we shall (informally) mean either: (a) Consensus
  protocols which are defined in the permissioned setting in order to
  deal with byzantine faults, or (b) Consensus protocols which work in
  the permissionless setting by importing protocols of type (a).} This
means that users are selected, and asked to carry out a consensus
protocol designed for the permissioned setting. So users are not only
asked to produce blocks, but also other objects, such as \emph{votes}
on blocks. Algorand is also a PoS protocol. 

Such fundamental differences between competing blockchain
protocols have hindered ``apples-to-apples'' comparisons between them,
and a majority of the research-to-date has focused on the analysis of
specific protocols (or narrow classes of protocols).  Our goal here is
to complement existing work on protocol-specific analysis with a
mathematical framework for formally discussing the relative merits of
protocols of very different types.

The first and main aim of this paper is to develop a model for the
analysis of permissionless blockchain protocols that blackboxes the
precise mechanics of the user selection process, allowing us to
isolate the properties of the selection process that are significant,
and to make comparisons between blockchain protocols of different
types.  Section \ref{setup} describes a framework of this kind,
according to which protocols run relative to a \emph{resource
  pool}. This resource pool specifies a balance for each user over the
duration of the protocol execution (such as hashrate or stake), which
may be used in determining which users are permitted to make
publications updating the state.

With this framework in place,  we then turn our
attention to consider how properties of the resource pool may
influence the functionality of the resulting protocol. In Sections \ref{A+F} and \ref{imp}, we will be
concerned with the distinction between scenarios in which the size of
the resource pool is known (e.g.\ PoS), and scenarios where the size
of the resource pool is unknown (e.g. PoW). 
We refer to these as the {\em sized} and {\em unsized} settings,
respectively.
We will find that the choice of setting is intimately related to a
fundamental tradeoff for permissionless blockchain protocols,
which can be viewed as an analog of the CAP theorem
from distributed computing \cite{gilbert2002brewer} for our framework:
In a partially synchronous setting,
a protocol cannot deliver \emph{finality} for block confirmations
while at the same time being \emph{adaptive}, in the sense that it
remains live without knowledge of the size of the resource
pool.\footnote{Our impossibility result contrasts with previous works
that prove positive results about the liveness and consistency
properties of the Bitcoin protocol in more strongly synchronous
settings (such as synchronous networks~\cite{GKL15} and networks with
bounded message delays~\cite{GKL15,SZ15,LRS18,WHGSW16}).}

In Section \ref{multi} we will examine a fundamental distinction
between PoW and non-PoW protocols, which concerns the forms of
permission that users are given to carry out updates to the state. We
formalize the idea that, under quite general
conditions, PoW protocols are distinguished by their ability to allow
the broadcast of specific blocks (and other objects), rather than
granting permission to broadcast any object from a large class (such as
any valid block extending a given position in the
blockchain). Historically, this is one of the distinctions between
PoW and PoS that has received the most attention in the literature
\cite{brown2019formal}.

\subsection{Finality and Adaptivity} \label{F+A_intro}

Our main impossibility result, which appears in Section \ref{imp}, concerns notions of `finality', `adaptivity', `security', and `liveness'.  
Before defining these terms,
it is useful to consider how these
notions relate to another informal division that is often drawn in the
cryptocurrency community and in the literature \cite{brown2019formal}, between
`longest chain' type protocols such as Bitcoin and Snow White on the one hand, and so
called `BFT' protocols such as Algorand \cite{chen2016algorand} and Tendermint
\cite{buchman2018latest} on the other.

Roughly speaking, the term `longest chain' is normally applied to
protocols which are derived from Bitcoin, and which work by having
miners select a fork of the blockchain to build on, which is defined
in terms of some sort of scoring function for chains. The selected
chain might be the one with the most PoW attached, or it might be
the longest, or it might be defined by an inductive process that
counts the number of descendants, as in the GHOST
protocol \cite{sompolinsky2015secure}. BFT protocols, on the other
hand, work by selecting a subset of users and having them carry out
a more traditional consensus protocol which is defined for the permissioned
setting.  The terms `BFT' and `longest chain' are thus descriptive of
\emph{where} protocols come from, but don't yet formally define
classes of protocols in a way that allows us to analyse the
differences between them and prove results contrasting the
performance of these classes of protocols in a broad sense.

The informal idea is that there is a trade-off.
While BFT protocols potentially offer \emph{finality} (whatever that
should mean), this comes with the price that the protocol will stall
if participation levels drop. In  Algorand, for example, committees of
users are selected in rounds, and block confirmation requires a
certain proportion of committee members to contribute signatures. If
participation levels drop to a point where insufficiently many
signatures are being produced for each block, then the process of
block confirmation will come to a standstill.  Longest chain protocols
such as Bitcoin, on the other hand, do not deliver finality but are
\emph{adaptive}, in the sense that they naturally adjust and remain
live in the face of fluctuating levels of participation.

In order to make this trade-off precise, we must decide how to
formalize `finality' and `adaptivity'.  First, let us consider
finality.  
To define this notion, we focus 
on differentiating the settings (e.g., the degree of assumed
synchrony) in which protocols are \emph{secure}, meaning that block
confirmation can be relied on.  Bitcoin and most other longest
chain-type protocols will not be considered to have finality, because
block confirmation can only be relied on assuming the ongoing
participation of a large fraction of honest users: Block confirmation
is not secure against unbounded network partitions \cite{WHGSW16}. Algorand and most
BFT-type protocols will have finality because there is essentially
zero chance that confirmed blocks will be rolled back, even in the
event that honest users cease to participate after the block has been
confirmed, e.g.\ due to an extended period of network failure.
Section \ref{A+F} formalizes this distinction and our finality notion
in terms of security in a \emph{partially synchronous setting}.  The
distinction between synchronous and partially synchronous settings is
standard when working with permissioned protocols~\cite{DLS88};
see Section \ref{setup} for the definitions.

Next, let us consider adaptivity. Recall that, according to our
framework, protocols run relative to a resource pool. This resource
pool specifies a balance for each user over the duration of the
protocol execution (such as hashrate or stake), which may be used in
determining which users are permitted to make publications updating
the state. So Bitcoin will be modeled as running relative to a
resource pool that specifies the hashrate of each user over the
duration of the protocol execution, and users with greater hashrate
will be more likely to be selected for block production. In Section
\ref{A+F}, we will formally define adaptivity in terms of the
information about the resource pool that is available to the protocol.
First, we will define the \emph{unsized} setting, so as to formalize
contexts in which the total size of the resource pool is information
which is not available to the protocol. With this definition in place,
and once we have formalized the notion of `liveness' -- roughly,
liveness is the property that with high probability the set of
confirmed blocks will grow over time -- we will then be able to define
adaptive protocols as those which are live in the unsized
setting.\footnote{One might want a protocol to satisfy stronger
  notions of liveness, such as quantitative lower bounds on chain
  growth or chain quality (as in e.g.~\cite{GKL15,WHGSW16}).  The fact
  that we work with such a weak notion of liveness only strengthens
  our impossibility result.}
Adaptive protocols are thus those which are like
Bitcoin---not necessarily in an operational sense, but in
the sense that they are live even when the total size of the resource
pool is unknown
to the protocol.

\subsection{The tradeoff between adaptivity and finality.}  \label{introtrade} 

With these definitions in place, we will then be able to formally
prove Theorem \ref{ourAmazingTheorem} below, which can be seen as an
analog of the CAP Theorem \cite{gilbert2002brewer} from distributed
computing for our framework. Roughly speaking, `security' in our
framework corresponds to `atomic consistency' in the framework in
which the CAP Theorem is proved in \cite{gilbert2002brewer}, and
`liveness' corresponds to `availability'. These correspondences are
not exact, however. While availability requires a response even
during extended periods of asynchrony, our definition of liveness
explicitly rules out the requirement that new confirmed blocks should
be produced under such conditions. The key observation in the proof of
Theorem \ref{ourAmazingTheorem} is that, in the unsized setting,
extended periods of asynchrony cannot be distinguished from a waning
resource pool. Liveness therefore forces the production of new
confirmed blocks during appropriately chosen periods of
asynchrony. Liveness and security are thus incompatible in the
partially synchronous and unsized setting, while the same is not true
in the partially synchronous and sized setting.

\begin{reptheorem}{ourAmazingTheorem}[Impossibility Result]
No protocol is both adaptive and has finality.
\end{reptheorem}

This result establishes a simple dichotomy for permissionless
blockchain protocols. A protocol can be adaptive or it can have
finality, but not both. It also draws a clean and formal line between
longest chain protocols such as Bitcoin and Ethereum
\cite{buterin2018ethereum}, or PoS implementations such as Snow White
\cite{bentov2016snow} on the one hand, and BFT protocols such as
Algorand, Casper FFG and PoS implementations of Tendermint or Hotstuff
\cite{yin2018hotstuff} on the other. While the former group are all
adaptive, the latter group all have finality.

Another interesting conclusion that can be drawn from Theorem
\ref{ourAmazingTheorem} concerns PoW protocols. 
PoS protocols are generally best modeled using the sized setting, while
PoW protocols are generally best modeled using the unsized
setting---the total stake is typically information which is available
to a protocol from the beginning of its execution, while the amount
of computational power used to provide PoW can vary over time in an
unpredictable way.  
To the extent that PoW protocols must operate in an unsized setting
(and guarantee liveness),
Theorem \ref{ourAmazingTheorem} implies that they cannot have
finality.\footnote{The continual adjustment of the difficulty 
parameter in Bitcoin can be viewed as an attempt to maintain an
approximation of the sized setting in a fundamentally unsized
setting.  See Section~\ref{ss:dfm} for further discussion.}

\subsection{Related Work}

The novel feature of Bitcoin that distinguishes it from previous
consensus protocols is that it is \emph{permissionless}, i.e.\ it
establishes consensus between a set of users that anybody can join,
with as many identities as they choose in any given role. This paper
can be seen as a step towards developing a formal
framework for the analysis of permissionless protocols akin to the
extensively developed one for permissioned protocols 
\cite{lynch1996distributed}. The study of byzantine fault
tolerant (BFT) consensus protocols in the permissioned setting dates
back at least to 1980 \cite{pease1980reaching,lamport2019byzantine}.
Among those BFT protocols of interest to us here, we can distinguish
two forms:

\begin{enumerate} 

\item The oldest relevant form of BFT protocol is aimed at solving the `Byzantine Generals Problem' \cite{pease1980reaching,lamport2019byzantine}. The task here is to reach consensus on a single yes/no decision. In applying these methods to the blockchain setting, one approach, as employed by Algorand, is to run such a BFT protocol for each in a sequence of proposed blocks, until consensus is reached for each block as to whether it should  be included in the blockchain. A drawback of this approach is that positive or negative consensus has to be reached for one block at a time. A large number of rounds of communication  may be required before agreement is reached, giving a corresponding negative impact on confirmation times. 

\item Perhaps more appropriate for the blockchain setting, since they are designed to achieve essentially the same task as permissionless blockchain protocols but in the permissioned setting, are BFT protocols designed for the purpose of 
state machine replication (SMR) \cite{castro1999practical,castro2002practical}. The task of such protocols is for a set of distributed users to agree on an order of execution for an ever growing list of client-initiated service commands -- replace `client-initiated service commands' with `transactions', and this is precisely the aim of permissionless blockchain  protocols. 
The advantage of this approach is that it allows for considerably simpler protocols, which might only require two or three rounds of communication per block.

\end{enumerate} 

The CAP theorem is one of the most celebrated theorems in the
distributed computing literature. The theorem proved by Gilbert and
Lynch \cite{gilbert2002brewer} is a formal version of a conjecture due
to Brewer, which is made in the context of distributed web
services. The theorem establishes an impossibility result: It is
impossible for such a distributed service to simultaneously achieve
the three desirable properties of \emph{consistency, availability},
and \emph{partition tolerance}. For formal definitions of these terms
we refer the reader to \cite{gilbert2002brewer} and
\cite{lynch1996distributed}. The relationship between the CAP Theorem
and Theorem \ref{ourAmazingTheorem} was briefly discussed in Section
\ref{introtrade}, and we shall expand on this discussion in Section
\ref{imp}.

\section{The framework}  \label{setup} 

\subsection{Predetermined and undetermined variables} 

Our aim here is to establish a framework for analysing permissionless
blockchain protocols that blackboxes the precise mechanics of the user
selection process. This will allow us to prove impossibility results,
and to isolate the properties of the selection process that are
significant, in the sense that they impact the way in which the
protocol must be designed, or influence properties of the resulting
protocol (such as security in a range of settings).
 
In order to define properties such as liveness and security later on,
it will be convenient to consider protocols that are specified
relative to a finite set of initially defined 
\emph{parameters}. 
For Algorand to run
securely, for example, one must first decide how long the protocol is
to run for, and then choose committee sizes accordingly. The duration
of the execution is therefore required as a
parameter
of the
protocol. Variables that are specified before the execution of the protocol as 
parameters, or which
take the same value for all executions of the protocol, are referred
to as 
\emph{predetermined}. 
Variables (such as the number of users)
that are not predetermined, will be referred to as
\emph{undetermined}.

\subsection{The users}  \label{users}

We consider settings in which protocols are executed by an undetermined
 set of pseudonymous users, this set being of undetermined
size.  Each user is given access to a signature scheme, and controls a
set of public keys by which they will be known to other users. We use
the variable $U$ to range over users, while $\mathtt{U}$ will be used
to range over public keys -- each user may have many public keys.
Amongst all users, there is one who is distinguished as the
\emph{adversary} and who controls an undetermined set of public keys.
Users other than the adversary are referred to as \emph{honest}. 

We will suppose that each user is a deterministic computing device,
which has amongst the actions it can perform calls to certain oracles,
as well as certain external functionalities such as the ability to
broadcast messages. The protocol specifies an \emph{instruction set}, which is a program which is run
by every user, other than the adversary. The adversary can follow any
program of their choosing.

While we might think of the set of users as forming a network over
which messages can propagate, in order to keep things as simple as
possible, we shall not make the network explicit in our framework.
Users simply have the ability to \emph{broadcast} messages. Once a
message is broadcast by a public key belonging to a given user, it may
subsequently be \emph{delivered} to other users at different stages of
the execution.

\subsection{Network failures} 

We suppose that protocols are specified to run for a predetermined
sequence of timeslots, each timeslot being of predetermined
length. The appropriate length of these timeslots 
depends on the protocol to be modeled.
For many PoS protocols, an appropriate length is 
slightly more than the network latency, i.e.\ the time it takes a block to propagate the
network. (Thus, each user might carry out many instructions during a
single timeslot.)  We will see that PoW protocols 
might be better
modeled using very short timeslots. 
In any case, the sequence
of timeslots is called the \emph{duration} $\mathcal{D}$.  At the beginning
of each timeslot in the duration, broadcast messages may be delivered
to various users. The \emph{message state} relative to a given user
is the set of all broadcast messages which have been delivered to
them.  The message state for a given user is therefore monotonically
increasing over time.  In order to be broadcast, a message must be
\emph{valid}, meaning that it must have a certain structure and that
certain other conditions, expanded on below, are also satisfied.

For example, if modeling Bitcoin or Snow
White, a user's message state will be the set of (valid) blocks that
have been delivered to them. Thus the message state will not, in
general, be a single chain of blocks. For Algorand, a user's message state
will be all those messages which have been delivered to them, which
are either valid blocks, or else the signed messages of committee
members exchanged while reaching consensus on blocks.

It is standard in the distributed computing literature to consider a
variety of \emph{synchronous, partially synchronous}, or
\emph{asynchronous} settings, in which users may or may not have
clocks which are almost synchronised, or run at varying speeds, and
where message delivery might be reliable or subject to various forms
of failure.  For the sake of simplicity, we will suppose here that
users' clocks are synchronised -- while this might seem like a strong assumption, this only strengthens our impossibility result.  We will, though, allow for
periods of network failure, during which the adversary is able to
control message delivery. In order to formalize this, we will suppose
that the duration is divided into intervals that are labelled either \emph{synchronous} or
\emph{asynchronous}  (meaning that each timeslot is either
synchronous or asynchronous).  We will suppose that, during
synchronous intervals, message delivery time is probabilistically
distributed for each pair of users.
During asynchronous intervals, we
suppose that the adversary is able to interfere with message delivery
as they choose, i.e.\ the adversary can leave messages to be delivered
in a probabilistic fashion as normal, can cause undelivered messages
to be delivered early, or can stop messages being delivered at all for
the duration of the asynchronous interval.  (Though we can always
assume that a message broadcast by a user is in effect delivered
instantaneously to that user.)
  We then distinguish two
\emph{synchronicity settings}. In the \emph{synchronous} setting it is
assumed that there are no asynchronous intervals during the duration,
while in the \emph{partially synchronous} setting there may be
undetermined asynchronous intervals.

\subsection{The structure of the blockchain} \label{blockstructure}  

Amongst all broadcast messages, there is a distinguished set referred
to as \emph{blocks}, and one block which is referred to as the
\emph{genesis block}.  Unless it is the genesis block, each block $B$
has a unique \emph{parent} block $\texttt{Par}(B)$, which must be
uniquely specified within the block message. Each block is produced by
a single user, $\texttt{Miner}(B)$, but may contain other broadcast
messages which have been produced by other users.  No block can be
broadcast by $\mathtt{U}:=\texttt{Miner}(B)$ at a point strictly prior
to that at which its parent has been delivered to $\mathtt{U}$.
$\texttt{Par}(B)$ is defined to be an \emph{ancestor} of $B$, and all
of the ancestors of $\texttt{Par}(B)$ are also defined to be ancestors
of $B$. If $B$ is not the genesis block, then it must have the genesis
block as an ancestor. At any point during the duration, the set of
broadcast blocks thus forms a tree structure.  If $M$ is a message
state, then we shall say that it is \emph{downward closed} if it
contains the parents of all blocks in $M$. By a \emph{leaf} of $M$, we
shall mean a block in $M$ which is not a parent of any block in
$M$. If $M$ is downward closed and contains a single leaf, then we
shall say that $M$ is a \emph{chain}.

\subsection{The resource pool} 

Protocols are run relative to a (predetermined or undetermined)
\emph{resource pool}, which in the general case is a function
$\mathcal{R}: \mathcal{U} \times \mathcal{D} \times \mathcal{M}
\rightarrow \mathbb{R}_{\geq 0}$,
where $\mathcal{U}$ is the set of public keys, $\mathcal{D}$ is the
duration and $\mathcal{M}$ is the set of all possible message
states.  So $\mathcal{R}$ can be thought of as specifying the resource
balance of each user at each timeslot in the duration, possibly
relative to a given message state.  For a PoW protocol like Bitcoin,
the resource balance of each public key will be their (relevant)
computational power at the given timeslot (which is generally
independent of any
message state). For PoS protocols, such as
Snow White and Algorand, however, the resource balance will be fully
determined by `on-chain' information, i.e.\ information recorded in
the message state  
$M$.
Generally, a chain of blocks 
$C\subseteq M$ 
will first be selected. So $C$ might be the longest chain, or the
longest chain of blocks that have been approved by committee
members. Then 
$\mathcal{R}(\mathtt{U},t,M)$ 
will be some function of $\mathtt{U}$'s stake as recorded by the blocks in $C$.\footnote{The details here will depend on the specific protocol. It's standard to insist that a user has had stake in the currency recorded for a certain number of timeslots before they are allowed to produce blocks, for example. So $\mathtt{U}$'s resource balance might be their stake according to $C$ or some initial segment of $C$, or else 0 if $\mathtt{U}$ has not been recorded as having non-zero stake for sufficient time.}

  By the \emph{total resource
  balance} $\mathcal{T}$, we mean the sum of the resource balances of
all public keys;
this is the function
$\mathcal{T}:\mathcal{D}\times \mathcal{M} \rightarrow
\mathbb{R}_{\geq 0}$ defined by $\mathcal{T}(t,M) = \sum_{\mathtt{U}} \mathcal{R}(\mathtt{U},t,M)$.
It should be noted that the resource pool is
a variable, meaning that a given protocol may be expected to be live
and secure with respect to a range of resource pools.

\subsection{The sized and unsized settings} \label{sizedunsized}

Just as we considered two synchronicity settings earlier, we also
consider two \emph{resource settings}. The basic idea is that
in the \emph{sized} setting, the total resource balance is information
which is available to the protocol (and the permitter, as described in Section \ref{permitter}), while in the
\emph{unsized} setting it is not.  The precise details are as follows.

\vspace{0.3cm}
\noindent \textbf{The unsized setting}.  For the unsized setting,
$\mathcal{R}$ (and hence $\mathcal{T}$) is undetermined, with the only restrictions being:
\begin{enumerate} 
\item $\mathcal{R}$   will be a function from $ \mathcal{U} \times
  \mathcal{D} \times \mathcal{M}$ to $\mathbb{R}_{\geq 0}$ satisfying
  the requirement that, at all timeslots in the duration,  the total
  resource balance belongs to a fixed
interval $[I_0,I_1]$, where $I_0>0$ is sufficiently small and
$I_1 >I_0$ is sufficiently large.\footnote{We consider resource pools with range restricted to the fixed interval $[I_0,I_1]$ because it turns out to be an overly strong condition to require a protocol to be live without
 \emph{any} further conditions on the total resource balance, beyond
 the fact that it is a function to $\mathbb{R}_{\geq 0}$. We wish to be
 able to talk about Bitcoin as live in the unsized setting, for
 example, but liveness will certainly fail if $\mathcal{R}$ is the
constant 0 function, or if the total resource balance decreases
sufficiently quickly over time.}
\item There may also be bounds placed on the resource balance of the adversary. 
\end{enumerate} 

We shall refer to the set of all resource pools satisfying these
restrictions as the \emph{possible resource pools}, and in Section
\ref{A+F} we shall define a protocol to be live if it is live for all
possible resource pools.

\vspace{0.3cm}
\noindent  \textbf{The sized setting}. For the sized setting, the total resource balance $\mathcal{T}$ is a predetermined function  $\mathcal{T}:  \mathcal{D} \times \mathcal{M} \rightarrow \mathbb{R}_{\ge 0}$.

\vspace{0.3cm} \hspace{0.3cm} The basic idea is that PoS protocols
will generally be best modeled using the sized setting, while PoW
protocols are best modeled using the unsized setting, since one does
not know the total resource balance (e.g., total hashrate in each
timeslot) in advance. There are some nuanced considerations, however.
With a PoS protocol, for example, one might not be able
to predict accurately what percentage of the stake will actually come
online and broadcast as requested by the protocol. So there may be
situations in which it is appropriate to define the resource balance
in terms of the online or contributing stake, and where it should be
recognised that only partial information will be available concerning
the total resource balance.  Equally, there may be contexts in which
good bounds can be given on the total resource balance over the
duration for the PoW case. The example of Bitcoin will be discussed
further in Section \ref{discuss}.

\subsection{The permitter oracle} \label{permitter} 

In order to specify how the resource pool is to be used, we shall make use of the notion of a \emph{permitter oracle}. This is the most critical part of the model, and is the part that blackboxes user selection, since it is the permitter oracle that grants permission to broadcast valid messages. The permitter oracle  need not be implemented explicitly in the blockchain being modeled, and is a mathematical abstraction that allows for  the discussion and comparison of blockchains of very different types. It is designed to be as simple as possible, subject to this goal. 

As described in Section \ref{users}, we consider each user to be a
computing device with access to certain external oracles and
functionalities. At any given timeslot $t\in \mathcal{D}$, a user's
\emph{state} is entirely specified by the set of public keys they
control, the protocol parameters, their message state and the set of
permissions they have been given by the permitter oracle $\mathtt{O}$.  The
\emph{protocol} $\mathtt{P}=(\mathtt{I},\mathtt{O})$ is then a pair, where the \emph{instruction set} $\mathtt{I}$ is a set of deterministic and
efficiently computable instructions, which specifies precisely what
actions honest users should carry out at each timeslot, as a function of the
timeslot and their state at that timeslot. The instructions of the
protocol are therefore a function of the timeslot, the keys controlled
by the user, the protocol parameters, their message state, and the
set of permissions they have been given by the permitter.

One of the external functionalities each user has is the ability to
broadcast valid messages.  Amongst the conditions required for validity
is that the public key responsible for the broadcast has been given
\emph{permission} by the permitter oracle $\mathtt{O}$, which is an
oracle to which users have access. We thus suppose that users can make
`requests' to the permitter, of the form $(\mathtt{U},M,t',A)$, where
$\mathtt{U}$ is a public key under their control, $M$ is a possible
downward closed message state, $t'$ is a timeslot, and where $A$ is
some (possibly empty) extra data. Given a request of this form, the
permitter may then respond by giving them permission to broadcast
certain messages.  The response of the permitter to a request
$(\mathtt{U},M,t',A)$ will be assumed to be a probabilistic function
of the protocol parameters, the actual timeslot $t$, the previous
requests made by $\mathtt{U}$, the tuple $(\mathtt{U},M,t',A)$, and of
the user's resource level
$\mathcal{R}(\mathtt{U},t',M)$.\footnote{Another way to interpret
  these conditions is that the response to a request (or the
  probability distribution on that response) should be fully
  determined by information known to the user making the request -- in practice, users should be able to check for themselves whether or
  not they have permission to broadcast.}

The conditions we have described above 
non-trivially restrict what the
permitter can do.
For example, consider the unsized setting, and suppose that the
total resource pool (e.g., total hashrate) 
cannot be deduced from the protocol parameters, $t$,
previous requests made by $\mathtt{U}$, the tuple
$(\mathtt{U},M,t',A)$, and the user's resource level
$\mathcal{R}(\mathtt{U},t',M)$.  In this
case, the framework requires that the response of the permitter 
be independent of the total resource pool.  (Whereas in the sized
setting, if the total resource pool can be deduced from the broadcast
state~$M$, the permitter is not so constrained.)

As we shall discuss later, the form of the permission given by the
permitter might be permission to broadcast a specific message (such as
the data $A$ proposed by the user in their request), or it might be
permission to broadcast any number of messages satisfying certain
criteria (such as any block that extends the message state at a
given location).  In what follows we shall consider various settings,
depending on what assumptions can be made about the relationship
between the permitter and the resource pool.

It should be noted that the roles of the resource pool and the
permitter are different in the sense that, while the resource pool is
a variable (meaning that a given protocol may be expected to be live
and secure with respect to a range of resource pools), the permitter is part of the protocol description  (meaning that a protocol is only required to run relative to a specific permitter
oracle).

\subsection{Modeling simple PoW and PoS protocols}  \label{examples}

For concreteness, we next consider how some simple PoS and PoW
protocols can be modeled using our framework.  In this section, we will give a brief summary. Then in Appendix A, we give more in-depth examples for Bitcoin and for a `generic' longest chain PoS protocol.  We remind the reader
that our goal is not to literally model the step-by-step operation of
these protocols, but rather to replicate the essential properties of
their user selection mechanisms with a suitable choice of a permitter
oracle.

First, consider a PoW protocol like Bitcoin. To keep things simple,
we'll initially ignore Bitcoin's adjustable `difficulty parameter'
(i.e., how hard the PoW is to produce); We'll return to this point in
Section \ref{discuss} and Appendix A.  To model a simple PoW protocol of this form,
we can consider very short timeslots (say 1 second each, or even
shorter).  The resource level (i.e., hashrate) of a user in a given
timeslot is independent of the message state, so we can restrict
attention to resource pools
$\mathcal{R}: \mathcal{U} \times \mathcal{D} \rightarrow
\mathbb{R}_{\geq 0}$.
We interpret a user request $(\mathtt{U},M,t',A)$ in a timeslot~$t$ as
all of $\mathtt{U}$'s efforts during timeslot~$t$ to extend the
message state~$M$.\footnote{The parameter~$t'$ is ignored by the
  permitter, or equivalently is automatically interpreted as the current timeslot~$t$; the parameter~$t'$ is relevant only for PoS protocols.}  
(If a user submits more than one request during a timeslot, the
permitter ignores all but the first.) For
example, we can interpret $A$ as a choice and ordering of transactions
within a proposed block, along with a choice of predecessor, with the
understanding that the user will try as many different nonces as
possible during the timeslot.  The permitter then gives~$\texttt{U}$
permission to broadcast with probability proportional to
$\mathcal{R}(\mathtt{U},t)$ (so long as $A$ can be legally added to
$M$).\footnote{For example, with probability
  $\mathcal{R}(\mathtt{U},t)/M$, where $M$ is a large constant that
  depends on the assumed maximum hashrate~$I_1$ and the timeslot length.}
A notable feature of this permitter is that permission is granted 
for the broadcast of specific messages (i.e., a specific choice
of~$A$), rather than for 
a collection of messages meeting certain criteria.

There are various ways in which `standard' PoS selection processes can
work. Let us restrict ourselves, just for now and for the purposes of
this example, to considering protocols in which the only broadcast
messages are blocks, and let us consider a longest chain PoS protocol
which works as follows: For each broadcast chain $C$ and for all
timeslots in a set $T(C)$, the protocol being modeled selects
precisely \emph{one} public key who is permitted to produce blocks
extending $C$ (i.e.\ blocks whose parent is the unique leaf of $C$),
with the probability each public key is chosen being proportional to
their wealth as recorded in $C$.\footnote{Note that being permitted to
  broadcast a block is not the same as being instructed by the protocol
  to broadcast a block, and does not determine how other users will
  treat the block -- there are many contexts in which users might be
  able to produce valid blocks for which broadcast is not instructed
  by the protocol (e.g., a block extending a chain other than the
  longest one). A user may also be permitted to produce two valid
  blocks whose broadcast constitutes an overt deviation from the
  protocol, and which might be punished.}  In order to model a
protocol of this form, we can consider a (timeslot-independent)
resource pool
$\mathcal{R}: \mathcal{U} \times \mathcal{M} \rightarrow
\mathbb{R}_{\geq 0}$,
which takes the longest chain $C$ from $M$, and allocates to each
public key $\mathtt{U}$ their wealth according to $C$.\footnote{Note
  that in many PoS protocols the relevant balance is actually
  $\mathtt{U}$'s wealth according to some proper initial segment of
  $C$, and that in modeling such protocols one should adjust
  $\mathcal{R}$ accordingly. As mentioned earlier, it is also standard
  to insist that $\mathtt{U}$ has been recorded as a public key with
  non-zero stake for a minimum number of timeslots.}  Then we can
consider a permitter which chooses one public key $\mathtt{U}$ for
each chain $C$ and each timeslot $t'$ in $T(C)$, each public key
$\mathtt{U}$ being chosen with probability 
$\mathcal{R}(\mathtt{U},C)/\mathcal{T}(C)$.  (This is well defined
because the total resource pool~$\mathcal{T}$ is known to the protocol.)
That chosen public key $\mathtt{U}$ corresponding to $C$ and $t'$, is
then given permission to broadcast blocks extending $C$ whenever
$\mathtt{U}$ makes a request $(\mathtt{U},M,t',\emptyset)$ for which
$C$ is the longest chain in $M$. A notable feature
of this permitter is that the permission it gives is for the
broadcast of \emph{sets} of messages satisfying certain criteria,
i.e.\ when the permitter gives permission it is for any (otherwise
valid) block extending a given chain $C$.

To model a BFT PoS protocol, the basic approach will be very similar to that described for the longest chain PoS protocol above, except that certain other signed messages might be now required in $M$ (such as signed votes on blocks) before permission to broadcast is granted, and permission may now be given for the broadcast of messages other than blocks (such  as votes on blocks).

We refer the reader to Appendix A for two examples that are described in greater detail. 

\section{Adaptivity and Finality}  \label{A+F}

\subsection{The extended protocol and the meaning of probabilistic
  statements} \label{probs}

In order to define what it means for a protocol to be secure or live,
we first need a \emph{notion of confirmation} for blocks, which is a
function $ \mathtt{C}$ mapping any message state to a chain which is a
subset of that message state, in a manner which depends on an
initially defined parameter called the \emph{security parameter}
$\varepsilon \geq 0$.  
(E.g., in Bitcoin, one might consider a block confirmed if six blocks
follow it on the longest chain; changing the ``six'' to some other number
would yield a different notion of confirmation.)
The intuition behind $\varepsilon$ is that it
should upper bound the probability of false confirmation. Given any
message state, $ \mathtt{C}$ returns the set of confirmed blocks.

In Section \ref{permitter}, we stipulated that the protocol $\mathtt{P}=(\mathtt{I},\mathtt{O})$ is a 
 pair, where the instruction set $\mathtt{I}$ is a set of deterministic and
efficiently computable instructions specifying precisely what actions an honest user should carry out  at each timeslot,  and where $\mathtt{O}$ is the permitter. 
 In general, however, a protocol might only be considered to run relative to a specific notion of confirmation $\mathtt{C}$. We will refer to the triple $(\mathtt{I},\mathtt{O},\mathtt{C})$ as the \emph{extended protocol}.   Often we shall suppress explicit mention of $\mathtt{C}$, and assume it to be implicitly attached to a given protocol. We shall talk about a protocol being live, for example, when it is really the extended protocol to which the definition applies.

Each execution of the extended protocol is then entirely determined by: 
\begin{enumerate} 
\item The parameters;
\item The set of users and their public keys; 
\item An index specifying the program executed by the adversary; 
\item The resource pool (which may or may not be undetermined); 
\item The set of asynchronous timeslots;
\item Certain events on which probability distributions are established, including permitter responses and message delivery times. 
\end{enumerate} 

Generally, when we discuss an extended protocol, we shall do so within the context of a \emph{setting}, which constrains the set of possible choices for (1)-(6) above.  The setting might specify the probability distribution on delivery times, for example, and might restrict the set of resource pools to those in which the adversary is given a limited resource balance. When we make a probabilistic statement to the effect that a certain condition holds with at most/least a certain probability, this means that the probabilisitic bound holds for all possible values of (1)-(5) above that are not made explicit in the statement, and which are consistent with the setting.

\subsection{Liveness and adaptivity}   

In order to define liveness for a protocol  with a  notion of confirmation $\mathtt{C}$, let $|\mathtt{C}(M)|$ denote the number of blocks in $\mathtt{C}(M)$ for any message state $M$. For a given $U$, and timeslots $t_1<t_2$, let $M_i$ be  $U$'s message state at $t_i$. Let us say that $[t_1,t_2]$ is a \emph{growth interval} for $U$, if $|\mathtt{C}(M_2)|>|\mathtt{C}(M_1)|$.

\begin{definition} \label{live} 
A protocol is \textbf{live} if, for every choice of security parameter  $\varepsilon>0$, there exists $\ell_{\varepsilon}$ such that the following holds with probability at least $1-\varepsilon$ for any timeslots $t_1< t_2\in \mathcal{D}$, and for any user $U$:  If  $t_2-t_1\geq \ell_{\varepsilon}$ and $[t_1,t_2]$ is entirely synchronous, then $[t_1,t_2]$ is a growth interval for $U$.
\end{definition}

\noindent So, roughly speaking, a protocol is live if the number of
confirmed blocks can be relied on to grow over time during synchronous
intervals of sufficient length. Note also, that while Definition \ref{live} only refers explicitly to protocols, it is really the \emph{extended protocol} to which the definition applies.  In order to properly understand
Definition \ref{live}, we refer the reader to the conventions
concerning the meaning of probabilistic assertions that were described
in Section~\ref{probs}. Generally, assertions of liveness and security
will be made within the confines of a particular setting, which might
restrict the probability distribution on message delivery times, or
limit the resource balance of the adversary (but is otherwise
worst-case subject to these constraints).

In order to digest Definition \ref{live}, it is useful to understand
why it should be satisfied by a protocol like Bitcoin. 
Suppose we model Bitcoin in the unsized
and synchronous setting. According to Section \ref{sizedunsized}, this
means that we assume the existence of a fixed interval $[I_0,I_1]$
such that $I_0>0$, $I_1>I_0$, and such that 
the total resource balance always takes
values in $[I_0,I_1]$. Let us suppose that we model Bitcoin and the
permitter as discussed in Section \ref{examples} -- again, for the
sake of simplicity, we'll forget about the fact that Bitcoin makes
adjustments to the `difficulty'. Suppose also that, as part
of the setting, we assume:
\begin{enumerate}
\item[(A)] The adversary only ever controls a suitably small proportion of the total resource balance, and; 
\item[(B)] The probability distribution on the length of time for message delivery is such that the probability of delivery failure tends to 0 as the time after broadcast tends to $\infty$. 
\end{enumerate} 
In order for a block $B$ to be confirmed, $\mathtt{C}$ requires that
it should belong to the longest chain in $M$ and be followed by $x$
many blocks, where the value of $x$ is a function of the security
parameter $\varepsilon$ and the assumed restriction on the
adversary's resource balance.  Suppose that, at a given timeslot $t$,
$C$ is the longest chain seen by $U$.  Since we assume that
the total resource balance always belongs to $[I_0,I_1]$, this allows
us to find $\ell_{\varepsilon}^{\ast}$ (independent of $t$ and $C$)
such that the following holds with probability $>1-\varepsilon/2$:
Some honest user with public key $\mathtt{U}'$ is permitted to broadcast a new block at
a timeslot before $t+\ell_{\varepsilon}^{\ast}$, but after all blocks
in $C$ have been delivered to them. In order to specify the value
$\ell_{\epsilon}$ whose existence is required by Defintion \ref{live},
we can then define $\ell_{\varepsilon}> \ell_{\varepsilon}^{\ast}$
such that, with probability $>1-\varepsilon/2$, the block broadcast by
$\mathtt{U}'$ will be delivered to $U$ by timeslot
$t+\ell_{\varepsilon}$.  It then holds with probability
$>1-\varepsilon$ that the longest chain (and hence the number of
confirmed blocks) seen by $U$ at $t+\ell_{\varepsilon}$
is of length greater than $|C|$.

Now that we have defined liveness, we can also define adaptivity: 

\begin{definition} 
We define a protocol to be \textbf{adaptive} if it is live in the unsized setting. 
\end{definition} 

\subsection{Security and finality} 

Roughly speaking, \emph{security} requires that confirmed blocks
normally belong to the same chain. Let us say that two distinct blocks
are \emph{incompatible} if neither is an ancestor of the other, and
{\em compatible} otherwise.  If $B\in \mathtt{C}(M)$ where $M$ is the
message state of $U$ at time $t$, then we shall say that
$B$ is confirmed for $U$ at $t$.

\begin{definition} \label{secure} A protocol is \textbf{secure} if the following holds for every choice of security parameter  $\varepsilon>0$, for every $U_1,U_2$ and for all timeslots $t_1,t_2$ in the duration:  With probability $> 1-\varepsilon$, all blocks which are confirmed for $U_1$ at $t_1$ are compatible with all those which are  confirmed for  $U_2$ at $t_2$.   
\end{definition}

\begin{definition} 
A protocol has \textbf{finality} if it is secure in the partially synchronous  setting. 
\end{definition}

Note that BFT protocols such as Algorand are normally designed to 
have finality in this sense.
For Algorand, the duration and adversary resource bound are initially
specified as parameters, and then the protocol specifies committee
sizes and other quantities so that the probability two incompatible
blocks will ever be confirmed is less than $\varepsilon$.

\section{The Impossibility of Adaptivity and Finality} \label{imp} 

In Section \ref{permitter}, we didn't describe any conditions
requiring that the behaviour of the permitter \emph{must} be
influenced by the resource pool. The only assumption of this kind that
we shall make is stated below, and will be applied for both the sized
and unsized settings.

\begin{quote}
\textbf{No balance, no voice}:  No $\mathtt{U}$ will be given permission to broadcast messages in response to a request $(\mathtt{U},M,t,A)$ for which $\mathcal{R}(\mathtt{U},t,M)=0$. 
\end{quote}

Now that the framework and all required definitions are in place, we
can formally prove Theorem~\ref{ourAmazingTheorem}. 

\begin{thm} \label{ourAmazingTheorem}
No protocol is both adaptive and has finality.
\end{thm}

As stated previously, this theorem can be seen as an analog of the CAP
Theorem~\cite{gilbert2002brewer} from distributed computing for our
blockchain protocol analysis framework. Now that we have formally
defined adaptivity, finality, security, and liveness, it may be useful
to say a little more about the relationship to the CAP Theorem. While
the CAP Theorem asserts that (under the threat of unbounded network
partitions), no protocol can be both available and consistent, it is
possible for BFT protocols such as Algorand to be both live and secure
{\em in the partially synchronous setting}.  This is possible because
liveness is a fundamentally weaker property than availability:
Liveness does not require new confirmed blocks to be produced during
extended periods of asynchrony.  For example, Algorand is live, even
though block production may stop during network partitions.  The key
idea behind the proof of Theorem~\ref{ourAmazingTheorem} is that, in
the unsized (and partially synchronous) setting, this fundamental
difference disappears, with network partitions indistinguishable from
waning resource pools.  Liveness then forces the existence of growth
intervals during network partitions. In the unsized and partially
synchronous setting, security and liveness thus become incompatible,
just as consistency and availability are incompatible
according to the CAP Theorem.

\begin{proof}
  (of Theorem~\ref{ourAmazingTheorem}) The idea behind the proof 
can be summed up as follows. We consider executions of
  the protocol in which there are at least two users, both of which
  are honest, and who control public keys $\mathtt{U}_0$ and
  $\mathtt{U}_1$ respectively. Suppose that, in an execution of the
  protocol in the unsized and partially synchronous setting,
  $\mathtt{U}_0$ and $\mathtt{U}_1$ both have the same constant and
  non-zero resource balance, and that all other users have resource
  balance zero throughout the duration. According to the assumption
  `no balance, no vote', this means that $\mathtt{U}_0$ and
  $\mathtt{U}_1$ will be the only public keys which are able to
  broadcast messages. For as long as the adversary is able to prevent
  messages broadcast by each $\mathtt{U}_i$ from being delivered to
  $\mathtt{U}_{1-i}$ ($i\in \{ 0,1 \}$), the execution will be
  indistinguishable, as far as $\mathtt{U}_i$ is concerned, from one
  in which only $\mathtt{U}_i$ has the same constant and non-zero
  resource balance. The fact that the protocol is live means that,
  with high probability, $\mathtt{U}_0$ and $\mathtt{U}_1$ will see
  confirmed blocks within a bounded period of time.  
  The confirmed blocks for $\mathtt{U}_0$ will be incompatible with
  those for $\mathtt{U}_{1}$, so long as these confirmed blocks appear
  before any point at which a message broadcast by $\mathtt{U}_i$ has
  been delivered to $\mathtt{U}_{1-i}$ for some $i\in \{ 0,1 \}$. This
  contradicts security for the protocol in the partially synchronous
  setting.

To describe the argument in more detail, let $\mathtt{U}_0$ and  $\mathtt{U}_1$ be public keys controlled by different honest users. For a duration  $\mathcal{D}$ which is sufficiently long,  we consider three different resource pools:

\begin{enumerate} 
\item[$\mathcal{R}_0:$]  We let $\mathcal{R}_0$ assign the constant value $I>0$ to both $\mathtt{U}_0$ and $\mathtt{U}_1$ over the entire duration, while all other users are assigned the constant value 0. 
\item[$\mathcal{R}_1:$]  We let $\mathcal{R}_1$ assign the constant value $I$ to $\mathtt{U}_0$  over the entire duration, while all other users are assigned the constant value 0.
\item[$\mathcal{R}_2:$]  We let $\mathcal{R}_2$ assign the constant value $I$ to $\mathtt{U}_1$  over the entire duration, while all other users are assigned the constant value 0.  
\end{enumerate} 

We consider three different executions of the protocol with the same parameters, for the unsized setting in which the resource pool is an undetermined variable: 

\begin{enumerate} 
\item[$\mathtt{Ex}_0:$] Here $\mathcal{R}:= \mathcal{R}_0$. All timeslots are asynchronous and the adversary prevents the delivery of messages broadcast by $\mathtt{U}_i$ to the user controlling $\mathtt{U}_{1-i}$, for $i\in \{ 0,1 \}$.

\item[$\mathtt{Ex}_1:$]  Here $\mathcal{R}:= \mathcal{R}_1$, and we work in the synchronous setting (or in the partially synchronous setting, but without interference by the adversary). 
\item[$\mathtt{Ex}_2:$]  Here $\mathcal{R}:= \mathcal{R}_2$, and we work in the synchronous setting.
\end{enumerate} 

According to the assumption of `no balance, no voice', it follows that
only $\mathtt{U}_0$ and $\mathtt{U}_1$ will be able to broadcast
messages in any of these three executions.  
Our framework stipulates that the instructions of the protocol
for a given user at a given timeslot must be a deterministic function
of the protocol parameters, the timeslot, the keys controlled by the
user, their message state and the set of permissions they have been
given by the permitter (see Section~\ref{permitter}). 
It also stipulates that the response of the permitter to a request
$(\mathtt{U},M,t',A)$ is a probabilistic function of the protocol
parameters, the actual timeslot $t$, previous requests made by
$\mathtt{U}$, the request $(\mathtt{U},M,t',A)$, and the user's
resource level $\mathcal{R}(\mathtt{U},t',M)$.  It therefore follows
by induction on timeslots that, because the resource pool is
undetermined:
\begin{enumerate} 
\item[$(\dagger)$] For each $i\in \{ 0,1 \}$, and for  all timeslots in $\mathtt{Ex}_0$, the probability distribution on the state of the user controlling $\mathtt{U}_i$  is identical to  the corresponding distribution at the same timeslot in $\mathtt{Ex}_{1+i}$. 
\end{enumerate} 
If the protocol is adaptive, then it follows from Definition
\ref{live} that we can find a timeslot $t_0$ satisfying the following
condition: In both $\mathtt{Ex}_{1+i}$ ($i\in \{ 0,1 \}$), it holds
with probability $>3/4$ that there is at least one block which is
confirmed for $\mathtt{U}_i$ at $t_0$. By $(\dagger)$ it then holds
for $\mathtt{Ex}_0$, and for each $i\in \{0,1\}$, that with
probability $>3/4$ there is at least one block which is confirmed for
$\mathtt{U}_i$ at $t_0$. We stipulated in Section
\ref{blockstructure} that no block $B$ can be broadcast by
$\mathtt{U}:=\texttt{Miner}(B)$ at a point strictly prior to that at
which its parent has been delivered to $\mathtt{U}$. It follows that
in $\mathtt{Ex}_0$ all blocks which are confirmed for $\mathtt{U}_i$
must be incompatible with all blocks which are confirmed for
$\mathtt{U}_{1-i}$. The definition of security therefore fails to hold
for timeslot $t_0$, and with respect to the security parameter 1/2.
\end{proof}

\section{Proof-of-Stake Requires Multi-Permitters} \label{multi} 

One major difference between typical PoW and PoS longest-chain
protocols (e.g., Bitcoin vs.\ Snow White) is the order of operations
between a user choosing a proposed block to broadcast and learning
whether or not it has permission to broadcast.  In the dominant PoW
protocols, the proposed block is chosen first, and only then is
permission granted or denied; in typical longest-chain PoS protocols,
permission (to broadcast in a given timeslot at a given location) is
granted before the specific block to broadcast is chosen.  Is this
difference an artefact of the protocols developed thus far, or is
it a more fundamental distinction between
PoW and non-PoW protocols?  We next use our framework to
reason about this question.

Already from the simple examples
in Section \ref{examples}, one can see that the standard PoW and PoS
protocols are best modeled by permitters and resource pools with
quite different properties. The permitter which we described in
modeling the PoS case, for example, was able to ensure that a single
user would be given permission to extend a particular chain at a
particular timeslot, simply because it has access to the total
resource balance recorded by a given chain.  As alluded to above, another notable
difference is that the permitter we described for the PoW case gave
permission for the broadcast of specific messages, rather than for
\emph{sets} of messages satisfying certain criteria (and of size
larger than 1).  We shall refer to permitters of this type as
\emph{single-permitters}, as opposed to \emph{multi-permitters}.

A key factor in determining whether multi-permitting is inherent to non-PoW protocols is the number of possible blocks extending a given chain
$C$---by the `possible' extensions of a chain $C$, we mean those
blocks $B$ satisfying all conditions required for validity
other than being permitted by the permitter
oracle. If the number of possible extensions is large, while the
probability that the permitter gives permission for each is small,
then a user may be able to increase their probability of gaining
permission to broadcast a block by churning through as many requests as
possible. This means that the probability of success comes to depend
on computational power, rendering the protocol (at least partially)
PoW.

In order to see this more precisely, we need a precise way to
talk about the computational power of a user. So, for the purposes of
this discussion, let us say that the computational power of a user is
the number of requests they are capable of making to the permitter in
each timeslot. We'll denote the computational power of $\mathtt{U}$ by $X_{\mathtt{U}}$. In order to restrict to realistic scenarios, we'll suppose that there is some fixed upper bound $X_{\text{max}}$, for which we always have $X_{\mathtt{U}} \leq X_{\text{max}}$. 
Suppose that, at a given timeslot $t$, $C$ is the longest chain, and, for the sake of simplicity, suppose that $C$ has been seen by all users. 
Suppose further that the following conditions are satisfied:

\begin{enumerate}   
\item[$(\dagger_1)$] The permitter $\mathtt{O}$ is a single-permitter. More specifically, let $\Lambda$ be the set of requests of the form $(\mathtt{U},C,t,B)$, such that $B$ is a possible extension of $C$. There exists some $\lambda >0$, such that $\mathtt{O}$ will respond to each distinct request   in $\Lambda$ made during timeslot $t$, by giving permission to broadcast the  specific block $B$ with independent probability $\lambda \cdot \mathcal{R}(\mathtt{U},t,C)$.
\item[$(\dagger_2)$]  For some constant $\mathtt{Ext}_{\text{No}}$, there are $\mathtt{Ext}_{\text{No}}$ many possible extensions of $C$ for each $\mathtt{U}$. Each  $\mathtt{U}$ submits $\text{min}\{ X_{\mathtt{U}}, \mathtt{Ext}_{\text{No}} \}$ many  requests from $\Lambda$ (and only those) during timeslot $t$. 
\end{enumerate} 

\noindent  Let $p_{\mathtt{U}}$ be the probability that $\mathtt{U}$ is given permission to broadcast during timeslot $t$. In what follows, it will simplify calculations to consider what happens in the limit of the size of the network of users: We shall say that a given condition holds \emph{in the limit}, if it holds so long as $p_{\mathtt{U}}$ is sufficiently small for all $\mathtt{U}$. We'll say that one quantity $x$ is proportional to another quantity $y$ in the limit, if there exists some constant $c$ such that, for each $\epsilon >0$, $x/cy \in (1-\epsilon,1+\epsilon)$ in the limit.

 Proposition
\ref{PoWprop} below says that, when the number of possible extensions $ \mathtt{Ext}_{\text{No}}$ is larger than $X_{\text{max}}$, the
single-permitter $\mathtt{O}$ automatically gives rise to a PoW
protocol, since, in the limit,  the probability $\mathtt{U}$ is given permission to broadcast is then  proportional to $\mathtt{U}$'s computational
power. While Proposition \ref{PoWprop} works according to the specific
assumption that the permitter responds to each request with
independent probability, it should be clear that the basic principle
holds under much more general conditions.

\begin{proposition} \label{PoWprop}
Suppose that $(\dagger_1)$ and $(\dagger_2)$ above are satisfied, so that $\mathtt{O}$ is a single permitter, and each $\mathtt{U}$ submits $\text{min}\{ X_{\mathtt{U}}, \mathtt{Ext}_{\text{No}} \}$ many requests during timeslot $t$. Let $p_{\mathtt{U}}$ be the probability that $\mathtt{U}$ is given permission to broadcast during timeslot $t$. In the limit, $p_{\mathtt{U}}$ is proportional to 
$ \mathcal{R}(\mathtt{U},t,C) \cdot \text{min}\{ X_{\mathtt{U}}, \mathtt{Ext}_{\text{No}} \}$.  
\end{proposition} 

\begin{proof} 
Define $Y_{\mathtt{U}}:= \text{min}\{ X_{\mathtt{U}}, \mathtt{Ext}_{\text{No}} \} $, so that $\mathtt{U}$ makes $Y_{\mathtt{U}}$ many requests during timeslot $t$.  If $Y_{\mathtt{U}}=0$ then $p_{\mathtt{U}}=0$ and  $\mathcal{R}(\mathtt{U},t,C) \cdot  Y_{\mathtt{U}}=0$. So suppose otherwise.  Let $\lambda$ be as defined in $(\dagger_1)$. Then the probability that at least one of the $Y_{\mathtt{U}}$ many requests made by $\mathtt{U}$ results in permission to broadcast is $1-(1-\lambda \cdot \mathcal{R}(\mathtt{U},t,C) )^{Y_{\mathtt{U}}}   $. It therefore suffices to show that:
\[   \frac{1-(1-\lambda \cdot \mathcal{R}(\mathtt{U},t,C) )^{Y_{\mathtt{U}}} }{\lambda \cdot \mathcal{R}(\mathtt{U},t,C) \cdot Y_{\mathtt{U}} }\rightarrow 1\mbox{    in the limit.   } \]
This can be shown with a straightforward analysis.
Expanding out $(1-\lambda \cdot \mathcal{R}(\mathtt{U},t,C) )^{Y_{\mathtt{U}}}$:
\begin{eqnarray*} 
(1-\lambda \cdot \mathcal{R}(\mathtt{U},t,C))^{Y_{\mathtt{U}}} &=& 1- \lambda \cdot \mathcal{R}(\mathtt{U},t,C) \cdot Y_{\mathtt{U}} +\\
&& \frac{1}{2} Y_{\mathtt{U}}(Y_{\mathtt{U}}-1)(\lambda \cdot \mathcal{R}(\mathtt{U},t,C))^2-\\
&& \frac{1}{6}  Y_{\mathtt{U}}( Y_{\mathtt{U}}-1)( Y_{\mathtt{U}}-2)(\lambda \cdot \mathcal{R}(\mathtt{U},t,C))^3 +\cdots 
\end{eqnarray*}
\noindent We therefore have:
\begin{eqnarray*} \label{eq1} \frac{1-(1-\lambda \cdot \mathcal{R}(\mathtt{U},t,C))^{Y_{\mathtt{U}}}}{\lambda \cdot \mathcal{R}(\mathtt{U},t,C) \cdot Y_{\mathtt{U}}} &= &1-\frac{(Y_{\mathtt{U}}-1)\cdot \lambda \cdot \mathcal{R}(\mathtt{U},t,C)}{2} + \\
&& \frac{(Y_{\mathtt{U}}-1)(Y_{\mathtt{U}}-2)(\lambda \cdot \mathcal{R}(\mathtt{U},t,C))^2}{6} + \cdots.\end{eqnarray*}
\noindent Now, since $Y_{\mathtt{U}}>0$, we have $\lambda \cdot  \mathcal{R}(\mathtt{U},t,C)\leq p_{\mathtt{U}}$, so that $\lambda \cdot  \mathcal{R}(\mathtt{U},t,C)$  must tend to zero as $p_{\mathtt{U}}$ tends to 0. Since $Y_{\mathtt{U}}$ is always less than the fixed bound $X_{\text{max}}$, the r.h.s.\ of (\ref{eq1}) therefore tends to 1 in the limit. 
\end{proof} 

In a standard PoS protocol, for example, one usually runs the
lotteries choosing users to produce blocks by having users hash
their public key, or some signed message, together with the timeslot
identifier and a frequently updated `random seed'. If the resulting
hash (considered as a real number) is the lowest produced, or if it is
below a threshold that depends on their stake, then that user might be
allowed to produce the next block. 
If one wanted the permission to
broadcast to be block-specific, one \emph{could} require users to
enter each proposed block as an extra input to the hash. Doing so
would mean that users who intend to produce blocks are now
incentivised to churn through many different possibilities for the
block as entry to the hash. So the resulting protocol becomes a
PoS/PoW hybrid.

In principle, however, and in situations where less possibilities are
required for each block, one certainly can envisage protocols which
use single-permitters, and which could be implemented using PoS. As a
simplistic example, we might consider a protocol which is aimed at
recording the time of a particular event. At each in a sequence of
short timeslots, a single user might be selected and given permission
of one of two forms.  Either:
\begin{enumerate} 
\item[(a)] They are given permission  to broadcast a block recording that, ``The event has happened by this timeslot'', or;
\item[(b)] They are given permission to broadcast a block recording that,  ``The event is yet to take place''. 
\end{enumerate} 
\noindent  
To ensure single-permitting, the parent of the block should also be
specified as part of the permission (e.g., with permission being given for
different parents in some rotating fashion). Honest users are then
asked to broadcast the permitted block only in the case that the
information recorded by the block and all ancestors is correct. Such a
protocol can be implemented using PoS, and the small number of
possibilities for each block means that one can do so without 
degenerating into a PoW protocol.

\section{Discussion} \label{discuss} 

\subsection{Defining finality}

The term `finality' is sometimes used to mean the \emph{absolute} guarantee
that blocks of transactions will not be revoked once committed to the
blockchain with a suitable level of confirmation. We have defined a
different (probabilistic) notion of finality, and have argued that it
can be effectively applied to the categorisation and analysis of
blockchain protocols. It may be instructive, however, to further
examine whether the former informal notion -- let's
call it \emph{absolute finality} -- is likely to be useful for the analysis of blockchain protocols.

To make things concrete, let us consider the case of Algorand. For the
purposes of this discussion, all one needs to know about Algorand is
that block confirmation revolves around the selection of committees,
and that the protocol relies for its security on the idea that an
adversary with suitably bounded stake will never have a committee
majority.\footnote{In fact, never more than a third of any given
committee.} Under appropriate modeling assumptions, one can show
that the chance of the adversary gaining a committee majority at any
point during the predetermined duration of the protocol
is indeed negligible. Since the process of selecting users to be
committee members is probabilistic, however, it certainly is
\emph{possible} that there will exist committees controlled entirely
by the adversary.  At a given moment in time it \emph{could} turn out
to be the case, even if only with negligible probability, that a
number of prior committees have actually had dishonest majorities, and
are now providing confirmation for an alternative
blockchain. So Algorand fails to have absolute finality as
a simple consequence of the fact that certain aspects of the process
are best modeled as probabilistic.

The question then becomes, 
is it meaningful in a blockchain context to worry about the
distinction between an event which occurs with probability which is
\emph{essentially} 0, and an event which holds with probability
\emph{exactly} 0?  While it might be possible for a committee to have
a dishonest majority, how much does this matter if the probability is
$<10^{-10}$ that this occurs at any time during the execution of the
protocol? We take the position that if a permissionless protocol
achieves absolute finality given appropriate modeling assumptions
(such as the security of elliptic curve cryptography, or the fact that
a given hash function is collision resistant), then it still holds
with non-zero probability  
that some aspect of the
modeling assumptions fails to hold. So the distinction is really a
matter of where one hides the probability of failure.

\subsection{Does finality matter?}\label{ss:dfm}

The extent to which protocol finality is important is an interesting
question.  We have defined finality here so as to be most useful for
classification purposes.  The notion of finality that we consider
requires being secure in the face of unbounded periods of network
failure; one might argue that this is overkill in practice.  For
example, one relaxation would require only that a protocol be
secure in the face of \emph{realistically bounded} periods of network
failure; this, in turn, may allow for greater protocol adaptivity.

Let us explore this idea further in the context of Bitcoin. 
In Section \ref{examples}, we considered how to model a PoW protocol
that was a simplified version of Bitcoin, in the sense that we did not
consider the updates to the `difficulty parameter' that are
implemented every couple of weeks in Bitcoin. Now that we have
formally defined security and adaptivity, we can consider in more
detail what differences are caused by these updates to the difficulty
parameter. In fact, Bitcoin is normally considered to be executed with
a notion of confirmation which is particularly insensitive to the
difficulty parameter -- a block is considered confirmed once it
belongs to the longest chain and is followed by a fixed number of
blocks (six being a common choice). According to this notion of
confirmation, network partitions of a few hours may suffice to produce
a situation in which different blocks are confirmed for different
users.  If one wants to avoid this, one response is
to consider the same
protocol paired with a notion of confirmation that requires blocks to
be produced at a certain \emph{rate}. For example, one might consider
a block to be confirmed if it belongs to the longest chain and is
followed by $x\geq 6$ many blocks, which have been produced in less
than $x/5.5$ many hours. The Bitcoin protocol with this notion of
confirmation is still adaptive, but the network partitioning attack
described in the proof of Theorem \ref{ourAmazingTheorem} would now
have to be carried out over a considerably extended interval of
time. One might argue that such extended network partitions are
unlikely, and that, realistically speaking, adaptivity (even if slow)
is likely to be beneficial in ensuring liveness.

\section{Concluding remarks} \label{conc}

Our main aim in this paper has been to establish a framework for
analysing permissionless blockchain protocols that blackboxes the
precise mechanics of the user selection process. Establishing such a
framework allows us to prove impossibility results, and to isolate the
properties of the selection process which are significant in the sense
that they impact the way in which the protocol must be designed, or
influence properties of the resulting protocol, such as security in a
range of settings. We have focussed on the difference between the
sized and unsized settings, and have shown that the choice of setting
is intimately related to a fundamental tradeoff for cryptocurrency
protocols: A protocol cannot deliver finality for block confirmations
while at the same time being adaptive. The formal dichotomy which
results can be seen as elucidating the informal division of
permissionless blockchain protocols into those which are longest chain
type protocols such as Bitcoin on the one hand, and those protocols
such as Algorand, Casper FFG \cite{buterin2017casper} or proof-of-stake (PoS) implementations
of Tendermint or Hotstuff on the other, which work by importing
traditional Byzantine-Fault-Tolerant protocols from the permissioned
to the permissionless setting.

In the description of the framework presented here, explicit mention was made of an adversary who displays byzantine behaviour.  The expectation is that properties of protocols are asserted modulo the existence of a bounded adversary. So assertions of liveness and security are made in a setting with explicit bounds on the adversary, and the requirement is that the protocol should behave well irrespective of the behaviour of the adversary, within the given bounds.  This is an entirely standard form of analysis in the distributed computing literature. There is a general understanding in the blockchain community, however, that in the blockchain setting there is also the need for a deeper game-theoretic analysis, which  takes account of user incentives. It is not enough that the protocol should perform well given arbitrary behaviour from the adversary. Given arbitrary behaviour by the adversary, it should also be the case that the instructions of the protocol constitute something like a Nash equilibrium for the honest users. It would be interesting to use and expand our framework in order to achieve impossibility results along these lines.

As well as allowing for impossibility results, a benefit of our framework may also be in providing some modularity for the description and analysis of protocols. For example, the description of PoS protocols tends really to consist of two components. One has to describe how lotteries are to be implemented securely, so as to provide an appropriate mechanism for user selection, and then one has to describe the protocol to be carried out by users who are selected to update the state. Once a mechanism for orchestrating lotteries has been agreed on (such as that used in Algorand), one might then want to describe a range of protocols, which work very differently from each other, but which use the same basic method of user selection. Or one might want to describe a protocol that uses the same method of user selection as Algorand, but which could be updated to use another method of user selection should something superior be developed later. Blackboxing the process of user selection via the use of permitters may therefore allow for a more modular description and analysis. 
   
A further avenue for research would be to use the
framework we have described here to formalize another notable
difference between protocols which are adaptive and protocols which
have finality, which concerns the nature of `proof of
confirmation'. For BFT protocols, it will generally be the case that
the very existence of a certain set of signed objects may suffice to
establish confirmation with high probability. For example, in Algorand, the existence of a block, together with an appropriate set of committee signatures establishing consensus for inclusion of the block, is sufficient to prove beyond reasonable doubt that the block can be considered confirmed. For Bitcoin and other
adaptive protocols, on the other hand, a user will only believe that a
certain chain is the longest until they are shown a longer chain. For
the adaptive protocols, in other words, one needs to see a user's full
message state in order to know whether they consider a given block
to be confirmed. For protocols with finality, by contrast, certain sets of publications will constitute
proof of confirmation, simply by virtue of being a \emph{subset} of a
user's state. We suspect that there are interesting interactions with
the resource setting to be explored in this regard.




\section{Appendix A -- Two Modeling Examples}  

In Section \ref{examples} we already considered how to model PoW and PoS protocols. The idea of this appendix is just to flesh those details out a little. In particular, we previously ignored Bitcoin's adjustable difficulty parameter, and we will now drop that simplifying assumption. 

\subsection{Modeling Bitcoin}  

We will assume that the reader is entirely familiar with the Bitcoin protocol, the conditions for block validity, and so on. In order to decide how to model Bitcoin, we just have to specify how the timeslots, the resource pool $\mathcal{R}$ and the permitter $\mathtt{O}$ are defined. 
We consider these in order.  

To model Bitcoin, we can use very short timeslots (say 1 second each, or even
shorter). The exact length of timeslots does not matter very much, but we will be assuming that each public key only attempts to mine at most one block in each timeslot. So the shorter timeslots are, the weaker this assumption becomes. We also want timeslots to be sufficiently short that any given miner is unlikely to produce a block in any single given timeslot. For the sake of concreteness,  let us  fix timeslots at 1 second each. Then defining the resource pool $\mathcal{R}$ is also simple.   The resource level (i.e., hashrate) of a user in a given
timeslot is independent of the message state, so we can restrict
attention to resource pools
$\mathcal{R}: \mathcal{U} \times \mathcal{D} \rightarrow
\mathbb{R}_{\geq 0}$. The value $\mathcal{R}(\mathtt{U},t)$ is the hashrate (i.e.\ hashes per second) of the public key $\mathtt{U}$ during timeslot $t$, i.e.\ the number of hashes per second that the user $U$, who controls the key $\mathtt{U}$, executes in attempting to mine a block with $\mathtt{U}$ specified as the miner. 

Our main task is therefore to determine how the permitter functions.
As described in Section \ref{examples}, we interpret a user request $(\mathtt{U},M,t,A)$ made during timeslot~$t$ as
all of $\mathtt{U}$'s efforts during timeslot~$t$ to mine a new block (or, rather, all of the efforts that the owner $U$ makes on behalf of the public key  $\mathtt{U}$). If $\mathtt{U}$ submits more than one request during a timeslot, the
permitter ignores all but the first. In order for the permitter to give a positive response, we suppose that $A$ must be an otherwise valid\footnote{I.e.,\ valid in all senses except that it has not yet been permitted for broadcast by $\mathtt{O}$.} block extending the longest chain in $M$. 
If this condition is satisfied then the \emph{probability} that the permitter gives a positive response will depend on $\mathtt{U}$'s hashrate, but will also depend on the adjustable  difficulty level. We  model the adjustable  difficulty level with a real valued function $p(M)$. So upon receiving the request above,  and when all other conditions for a positive response that we have already listed are satisfied, $\mathtt{O}$ now gives permission to broadcast $A$ with probability: 

 \[ \text{min} \{ p(M)\cdot \mathcal{R}(\mathtt{U},t),1 \}.\] 
 
\noindent  In order to complete our description of the permitter, it
therefore remains to specify how $p(M)$ is determined. Of course, this
has to work in essentially the same way as Bitcoin: $p(M)$ will be the
last in a sequence of real number values $p_1,p_2,\dots $ that is
updated every 2016 blocks along the longest chain in $M$. Recall that
the difficulty level is adjusted to try and maintain block
production at a rate of one block every 10 minutes ($=600$ seconds) on
average and that, initially, the difficulty level in Bitcoin
required a hash ending with 32 zeros. So we start with $p_1=
\frac{1}{600 \cdot 2^{32}}$. Every 2016 blocks (working along the
longest chain in $M$), a production time $T_i$ (in seconds) for the  $i$th sequence of  2016 blocks is determined from the block timestamps.   Then, subject to certain caveats listed below,  we define: 
 
 \[ p_{i+1}= p_{i} \cdot \frac{T_i}{2016\times 600}. \]

\noindent The caveat to the definition of $p_{i+1}$ above is that, to stop the difficulty level changing too quickly, Bitcoin specifies that $p_{i+1}$ can change by at most a factor 4 at a time, i.e. that $p_{i+1}$ must belong to $[\frac{1}{4} p_i, 4p_i]$. So, in line with Bitcoin,  if the definition of $p_{i+1}$ above gives a value $x$ outside this interval $[\frac{1}{4} p_i, 4p_i]$, then we define $p_{i+1}$ to be whichever of $\frac{1}{4}p_i$ and $4p_i$ is closest to $x$. 

The reader may notice that the fixed interval $[I_0,I_1]$, which is described in Section \ref{sizedunsized} as being a significant part of the underlying assumptions for the unsized setting, does not explicitly feature in our description of the model for Bitcoin. The existence of this interval does become significant, however, once one tries proving liveness for the extended protocol.   

\subsection{Modeling a `generic' longest chain PoS protocol} 

Rather than deal with the idiosyncrasies of any particular well known  PoS  protocol, for the sake of simplicity we will consider how to model a `generic' longest chain PoS protocol, for which the only broadcast  messages are blocks. Since we already described roughly how to model protocols of this form in Section \ref{examples},  the point of this section is just to examine in more detail how choices in the protocol definition will be reflected in the model.

Most PoS protocols already consider explicit timeslots, and allow for the addition of one new block to the longest chain for each timeslot. So we will consider a protocol which comes with explicitly defined timeslots of this form, and we will assume that each block $B$ comes with a corresponding timeslot $t(B)$ -- we will also refer to $t(B)$ as the \emph{timestamp} for $B$. For the sake of concreteness we will suppose that each timeslot is 30 seconds long.   We assume that, at each timeslot, the protocol directs honest users to try and extend the longest chain. In order to determine our model, we are left to specify how $\mathcal{R}$ and the permitter $\mathtt{O}$ should be defined.

First of all, let us consider $\mathcal{R}$. Of course, the basic idea with a PoS protocol is that the resource pool should reflect a public key's stake in the currency. In order to model a
protocol of this form, we can therefore consider a timeslot-independent
resource pool
$\mathcal{R}: \mathcal{U} \times \mathcal{M} \rightarrow
\mathbb{R}_{\geq 0}$,
which takes the longest chain $C_M$ from $M$, and allocates to each
public key $\mathtt{U}$ a resource balance which is determined by the information recorded in $C_M$. 
Precisely how this resource balance should be determined from $C_M$ will, however, depend on the particular details of the protocol. Let $t(C_M)$ be the timestamp of the leaf of $C_M$. It is standard practice in PoS protocols to require that, if a user is to produce a block which extends $C_M$, then they should have a non-zero stake in the currency at some timestamp $t$ which is significantly less than $t(C_M)$. 
For the sake of concreteness, let us suppose that the protocol we are modelling considers the relevant balance to be that at timeslot $t^{\ast}:=\text{max} \{ t(C_M)-\text{ 1 hour}, 0 \}$. Then we define $\mathcal{R}(\mathtt{U},M)$ to be $\mathtt{U}$'s stake at timeslot $t^{\ast}$, as recorded in $C_M$.

Next, let us consider the permitter $\mathtt{O}$. Again, the precise details as to how we define $\mathtt{O}$ will depend on the protocol being modeled.  It is fairly common for PoS protocols to specify that blocks cannot have parents which are too much older than they are. So, for the sake of concreteness, let us suppose that the protocol we are modeling  requires that each block $B$ must have a parent whose timestamp is at most an hour earlier than $t(B)$ in order to be valid.  Define $T(C_M):= \{ t|\ t\in \left(t(C_M),t(C_M) + \text{ 1 hour} \right]\}$. Let us suppose that, for each $t\in T(C_M)$, the protocol being modeled selects
precisely one public key who is permitted to produce blocks
extending $C_M$ (i.e.\ blocks whose parent is the unique leaf of $C$ and with timestamp $t$),
with the probability each public key $\mathtt{U}$ is chosen being proportional to $\mathcal{R}(\mathtt{U},M)$.  In this case, we can
simply consider a permitter which chooses one public key $\mathtt{U}$ for
each chain $C$ and each timeslot $t\in T(C)$, each public key
$\mathtt{U}$ being chosen with probability 
$\mathcal{R}(\mathtt{U},C)/\mathcal{T}(C)$.  This is permissable
because the total resource pool~$\mathcal{T}$ is a predetermined variable.
That chosen public key $\mathtt{U}$ corresponding to $C$ and $t$, is
then given permission to broadcast blocks extending $C$ whenever
$\mathtt{U}$ makes a request $(\mathtt{U},M,t,\emptyset)$ for which
$C=C_M$, i.e., for which $C$ is the longest chain in $M$.

\end{document}